\documentclass[preprint,12pt]{elsarticle}

\usepackage{amssymb}
\usepackage{graphicx}
\usepackage{wrapfig}
\usepackage{latexsym}
\usepackage{amsmath}
\usepackage{algorithm}
\usepackage{algorithmic}
\usepackage{multicol}
\usepackage{amsthm}
\usepackage{subfigure}

\newtheorem{theorem}{Theorem}
\newtheorem{lemma}{Lemma}
\newtheorem{definition}{Definition}

\journal{Theoretical Computer Science}

\begin{document}

\begin{frontmatter}



\title{Partial Gathering of Mobile Agents \\in Asynchronous Unidirectional Rings\tnoteref{t1,t2}}
\tnotetext[t1]{The conference version of this paper is published in the proceedings of 16th International
Conference on Principles of Distributed Systems (OPODIS 2012).}
\tnotetext[t2]{This work was supported by JSPS KAKENHI Grant Numbers 24500039,
24650012, 25104516, 26280022, and 26330084.
}


\author{Masahiro~Shibata\corref{cor1}}
\ead{m-sibata@ist.osaka-u.ac.jp}
\author{Shinji~Kawai\corref{cor1}}

\author{Fukuhito~Ooshita\corref{cor1}}
\ead{f-oosita@ist.osaka-u.ac.jp}
\author{Hirotsugu~Kakugawa\corref{cor1}}
\ead{kakugawa@ist.osaka-u.ac.jp}
\author{Toshimitsu~Masuzawa\corref{cor1}}
\ead{masuzawa@ist.osaka-u.ac.jp}

\cortext[cor1]{Corresponding author. Tel.:+81 6 6879 4117. Fax: +81 6 6879 4119.}

\address{Graduate School of Information Science and Technology, Osaka University, 1-5
Yamadaoka, Suita, Osaka 565-0871, Japan}

\begin{abstract}
In this paper, we consider the partial gathering problem of mobile agents in
asynchronous unidirectional rings equipped with whiteboards on nodes.
The partial gathering problem is a new generalization of the total gathering problem.
The partial gathering problem requires, for a given integer  $g$, that each agent
should move to a node and terminate so that at least $g$ agents
should meet at the same node. 
The requirement for the partial gathering problem is weaker than that for
the (well-investigated) total gathering problem, and thus, we have interests in
clarifying the difference on the move complexity between them.
We propose three  algorithms to solve the partial gathering problem. 
The first algorithm is deterministic but requires unique ID of each agent. 
This algorithm achieves the partial gathering in $O(gn)$ total moves, 
where $n$ is the number of nodes. 
The second algorithm is randomized and requires no unique ID of each agent (i.e., anonymous).
This algorithm achieves the partial gathering in expected $O(gn)$  total moves. 
The third  algorithm is deterministic and requires no unique ID of each agent.
For this case, we show that there exist initial configurations in which no algorithm can solve the problem 
and  agents can achieve the partial gathering in $O(kn)$ total moves for solvable initial configurations,
where $k$ is the number of agents.
Note that the total gathering problem requires $\Omega (kn)$ total moves, while the partial gathering problem requires $\Omega (gn)$ total moves in each model.
Hence, we show that the move complexity
of the first and second  algorithms is asymptotically optimal.


\end{abstract}


\end{frontmatter}



\section{Introduction}
\label{intro}

\subsection{Background and our contribution}
\label{background}

A {\em distributed system} is a system that consists of a set of computers ({\em nodes})  and communication links.
In recent years, distributed systems have become large and design of distributed systems has become complicated.
As an effective way to design distributed systems, (mobile) agents have attracted a lot of attention \cite{gathering}.
Design of distributed systems can be simplified using agents because they can traverse the system with carrying information and process tasks on each node.

The total gathering problem is a fundamental problem for cooperation of agents \cite{gathering,token1,token2}. 
The total gathering problem requires all agents to meet at a single node in finite time. 
The total gathering problem is useful because, by meeting at a single node, all agents can share information or synchronize behaviors among them.

In this paper, we consider a new generalization of the total gathering problem, called the {\em partial gathering problem}. The partial gathering problem does not always require all agents to gather at a single node, but requires agents to gather partially at several nodes. More precisely, we consider the problem which requires, for a given integer  $g$, that each agent should move to a node and terminate at a node so that at least $g$ agents should meet at the node.
We define this problem as the {\em $g$-partial gathering problem}. 
We assume that $k$ is the number of agents. 
Clearly if $k/2 < g\leq k$ holds, the $g$-partial gathering problem is equivalent to the total gathering problem. If $g\leq k/2$ holds, the requirement for the $g$-partial gathering problem is weaker than that for the total gathering problem, and thus it seems possible to solve the $g$-partial gathering problem with fewer total moves. From a practical point of view, the $g$-partial gathering problem is still useful because agents can share information and process tasks cooperatively among at least $g$ agents.

The contribution of this paper is to clarify the difference on the move complexity between the total gathering problem and the $g$-partial gathering problem.
We investigate the difference in asynchronous unidirectional rings equipped with whiteboards on nodes.
The contribution of this paper is summarized in Table \ref{table:result}, where $n$ is  the number of nodes.
\begin{table*}[t]
\begin{center}
\small
\caption[smallcaption]{Proposed algorithms for the $g$-partial gathering problem in asynchronous unidirectional rings.}
\label{table:result}
\newlength{\myheight}
\setlength{\myheight}{10.5mm}
\begin{tabular}{|c|c|c|c|}
\hline \parbox[c][\myheight][c]{0cm}{} & \parbox[c]{18mm}{\shortstack {Model 1\\ (Section 3)}}  & 
\parbox[c]{18mm}{\shortstack{Model 2\\ (Section 4)}} & \parbox[c]{18mm}{\shortstack{Model 3\\(Section 5)}} \\
\hline Unique agent ID & Available & Not available & Not available \\
\hline \parbox[c][\myheight][c]{0cm}{} \parbox[c]{25mm}{Deterministic\\/Randomized}  & Deterministic & Randomized & Deterministic  \\
\hline Knowledge of $k$ & Not available & Available & Available \\
\hline The total moves & $O(gn)$ & $O(gn)$ & $O(kn)$  \\
\hline \parbox[c][\myheight][c]{0cm}{}Note     & -  & - & \parbox[c]{45mm}{\shortstack{There exist \\ unsolvable configurations}} \\
\hline 
\end{tabular}
\end{center}
\end{table*}
First, we propose a deterministic algorithm to solve the $g$-partial gathering problem for the case that agents have distinct IDs.
This algorithm requires $O(gn)$ total moves.
Second, we propose a randomized algorithm to solve the $g$-partial gathering problem for the case that agents have no IDs but agents know the number $k$ of agents.
This algorithm requires expected $O(gn)$ total moves.
Third, we consider a deterministic algorithm to solve the $g$-partial gathering problem for the case that agents have no IDs but agents know the number $k$ of agents.
In this case, we show that there exist initial configurations for which the $g$-partial gathering problem is unsolvable. 
Next, we propose a deterministic algorithm to solve the $g$-partial gathering problem for any solvable initial configuration.
This algorithm requires $O(kn)$ total moves.
Note that the total gathering problem requires $\Omega (kn)$ total moves regardless of deterministic or randomized settings.
This is because in the case that  all the agents are uniformly deployed, at least half agents require $O(n)$ moves to meet at one node.
Hence, the first and second  algorithms imply that the $g$-partial gathering problem can be solved with fewer total moves compared to the total gathering problem for the both cases.
In addition, we show a lower bound  $\Omega(gn)$ of the total moves  for the $g$-partial gathering problem if $g\ge 2$. This means the first and second algorithms are  asymptotically optimal in terms of the total moves.

\subsection{Related works}
\label{work}

Many fundamental problems for cooperation of mobile agents have been studied. For example, the searching problem \cite{black2}, the gossip problem \cite{Mrs.Suzuki}, the
election problem \cite{direction}, and the gathering problem \cite{gathering,token1,token2,direction,speed,token3, token4,black1,Mr.Kawai,LocationAware,Tree1,Tree2,Tree3,Arbitrary1} have been studied.

In particular, the gathering problem has received a lot of attention and has been extensively studied in many topologies including trees \cite{gathering,Mrs.Suzuki,LocationAware,Tree1,Tree2,Tree3}, tori \cite{gathering,token3}, and rings \cite{gathering,token1,token2,direction,speed,token4,black1,Mr.Kawai}. The gathering problem for rings has been extensively studied because
algorithms for such highly symmetric topologies give techniques to treat the essential difficulty of the gathering problem such as breaking symmetry.

For example, Kranakis et al. \cite{token1} considered  the gathering problem for  two mobile agents in ring networks.
This algorithm allows each agent to use a token to select the gathering node based on the token locations.
Later this work has been  extended to consider any number of agents  \cite{token2,speed}.
Flocchini et al. \cite{token2} showed that if one token is available for each agent, the lower bound on the space complexity per agent is $\Omega (\log k + \log \log n)$ bits, where $k$ is the number of agents and $n$ is the number of nodes.
Later, Gasieniec et al. \cite{speed} proposed the asymptotically space-optimal algorithm for uni-directional ring networks.
Barriere et al. \cite{direction} considered the relationship between the gathering problem and the leader agent election problem.
They showed that the gathering problem and the leader agent election problem are solvable under only the assumption that the ring has sense of direction and the numbers of nodes and agents are relatively prime.

A fault tolerant gathering problem is considered in \cite{token4, black2}.
Flocchini et al. \cite{token4} considered  the gathering problem when tokens fail and showed  that 
knowledge of $n$ (number of agents) allows better time complexity than knowledge of $k$ (number of agents).
Dobrev et al. \cite{black2} considered  the gathering problem for the case that there exists a dangerous node, called a black hole.
A black hole  destroys any agent that visits there.
They showed that it is impossible for all agents to gather and they considered how many agents can survive and gather.

A randomized algorithm to solve the gathering problem is shown in \cite{Mr.Kawai}.
Kawai et al. considered the gathering problem for multiple agents under the assumption that agents know neither $k$ nor $n$,
and proposed a randomized algorithm to solve the gathering problem with  high probability in $O(kn)$ total moves.

\subsection{Organization}
The paper is organized as follows. 
Section \ref{model} presents the system models and the problem to be solved.
In Section \ref{sec:deterministic} we consider the first model, that is, the algorithm is deterministic and each agent has a distinct ID.
In Section \ref{sec:randomized} we consider the second model, that is, the algorithm is randomized and agents are anonymous.
In Section \ref{sec:DeterAnonymous} we consider the third model, that is, the algorithm is deterministic and agents are anonymous.
Section \ref{conclusion} concludes the paper.

\section{Preliminaries}
\label{model}
\subsection{Network model}
\label{network}
A {\em unidirectional ring network} $R$ is a tuple $R = (V,L)$, where $V$ is a set of nodes and $L$ is a set of unidirectional communication links. We denote by $n$ ($=|V|$) the number of nodes. Then, ring $R$ is defined as follows.
\begin{itemize}
\item $V =\{v_0,v_1,\ldots ,v_{n-1}\}$ 
\item $L =\{(v_i,v_{(i+1)\bmod n})~|~0\leq i\leq n-1\}$ 
\end{itemize}
We define the direction from $v_i$ to $v_{i+1}$ as the {\em forward} direction, and the direction from $v_{i+1}$ to $v_{i}$ as the  {\em backward} direction.
In addition, we define the $i$-th ($i\neq 0$) forward (resp., backward) agent ${a_h'}$ of  agent $a_{h}$ as 
the agent that exists in the $a_h$'s forward (resp., backward) direction and there are $i-1$ agents between $a_h$ and $a_{h'}$.
Moreover, we call the $a_h$'s 1-st forward and backward agents \textit{neighboring agents} of $a_h$ respectively.

In this paper, we assume nodes are anonymous, i.e., they do not have IDs. Every node $v_i\in V$ has a whiteboard that  agents on node $v_i$ can read from and write on the whiteboard of $v_i$. 
We define $W$ as a set of all states (contents) of a whiteboard.
\subsection{Agent model}
\label{agent}
Let $A=\{a_1,a_2,\ldots ,a_k\}$ be a set of agents. We consider three model variants.

In the first model, we consider agents that are distinct (i.e., agents have distinct IDs) and execute a deterministic algorithm. 
We model an agent $a_h$ as a  finite automaton $(S,\delta,s_{initial},s_{final})$. 
The first element $S$ is the set of the $a_h$'s all states, which includes initial state $s_{initial}$ and final state $s_{final}$.
When $a_h$ changes its state to $s_{final}$, it terminates the algorithm.
The second element $\delta$ is the state transition function. Since we treat deterministic algorithms, $\delta$ is a mapping  $S\times W\rightarrow S\times W\times M$, 
where  $M=\{1,0\}$ represents whether the agent makes a movement or not in the step. The value 1 represents movement to the next node and 0 represents stay at the current node. 
Since rings are unidirectional, each agent moves only to its forward node.
Note that if the state of $a_h$  is $s_{final}$ and the state of its current node's whiteboard is $w_i$,
then $\delta \,(s_{final}, w_i)=(s_{final},w_i,0)$ holds.
In addition, we assume that each agent cannot detect
whether other agents exist at the current node or not.
Moreover, we assume that each agent knows neither the number of nodes $n$ nor agents $k$.
Notice that $S,\delta,s_{initial},$ and $s_{final}$ can be dependent on the agent's ID.

In the second model, we consider agents that are anonymous (i.e., agents have no IDs) and execute a randomized algorithm.
We model an agent similarly to the first model except for state transition function $\delta$. 
Since we treat randomized algorithms, $\delta$ is a mapping  $S\times W\times R\rightarrow S\times W\times M$, where $R$ represents a set of random values. 
Note that if the state of some agent is $s_{final}$ and the state of its current node's whiteboard is $w_i$,
then $\delta \,(s_{final},w_i,R)=(s_{final},w_i,0)$ holds. 
In addition, we assume that each agent cannot detect
whether other agents exist at the current node or not,
but  we assume that each agent knows the number of agents $k$. 
Notice that all the agents are modeled by the same state machine since they are anonymous.

In the third model, we consider agents that are anonymous and execute a deterministic algorithm.
We also model an agent similarly to the first model.
We assume that each agent knows the number of agents $k$. 
Note that  all the agents are modeled by the same state machine.

\subsection{System configuration}
\label{sec:configuration}
In an agent system, (global) {\em configuration} $c$ is defined as a product of
states of agents, states of nodes (whiteboards' contents), and locations of agents.
 We define $C$ as a set of all configurations. In initial configuration $c_0\in C$, we assume that no pair of agents stay at the same node. We assume that each node $v_j$ has boolean variable $v_j.initial$ at the whiteboard that indicates existence of agents in the initial configuration. 
If there exists an agent on  node $v_j$ in the initial configuration, the value of $v_j.initial$ is true. Otherwise, the value of $v_j.initial$ is false.

Let $A_i$ be an arbitrary non-empty set of agents. When configuration $c_i$ changes to $c_{i+1}$ by the step of every agent in $A_i$, we denote the transition by $c_i \xrightarrow{A_i} c_{i+1}$.
In one atomic step, each agent $a\in A_j$ executes the following series of events:
1) reads the contents of  its current  node's whiteboard, 2) executes local computation, 
3) updates the contents of  the node's whiteboard, and 4) moves to the next node or stays at the current node.
We assume that agents move instantaneously, that is, agents always exist at nodes (do not exist at links). 
If multiple agents at the same node are included in $A_i$, the agents take steps in an arbitrary order.
When $A_i = A$ holds for every $i$, all
agents take steps every time. This model is called the {\em synchronous model}.
Otherwise, the model is called the {\em asynchronous model}.
In this paper, we consider the asynchronous system.

If sequence of configurations $E = c_0,c_1,\ldots $ satisfies $c_i \xrightarrow{A_i} c_{i+1}$  $(i\geq 0)$, $E$ is called an $execution$ starting from  $c_0$ by schedule $A_1, A_2,\ldots $ .
We consider only fair schedules, where every agent appears infinitely often.
Execution $E$ is infinite, or ends in final configuration $c_{final}$ where every agent's state is $s_{final}$.

\subsection{Partial gathering problem}
\label{partial}

The requirement of the partial gathering problem is that, for a given integer  $g$, each agent should move to a node and terminate so that at least $g$ agents should meet
at every node where an agent terminates. Formally, we define the $g$-partial gathering problem as follows.

\begin{definition}\label{teigi:partial}
Execution $E$ solves the $g$-partial gathering problem when the following conditions hold:
\begin{itemize}
\item Execution $E$ is finite (i.e., all agents terminate).
\item In the final configuration, all agents are in the finial states, and for any node $v_j$ where an agent terminates, there exist at least $g$ agents on $v_j$. 

\end{itemize}
\end{definition}

For the $g$-partial gathering problem, we have the following lower bound on the total number of agent moves.
This lemma holds in both deterministic and randomized algorithms.

\begin{theorem}
\label{lower}

The total number of agent moves required to solve the $g$-partial gathering problem is   $\Omega (gn)$ if $g\ge 2$. 
\end{theorem}
\begin{proof}
We consider an initial configuration such that all agents are scattered evenly 
(i.e.,  all the agents have the same distances to their nearest agents). 
We assume $n=ck$ holds for some positive integer $c$. Let $V'$ be the set of nodes where agents exist in the final configuration, and let $x=|V'|$. Since at least $g$ agents meet at $v_j$ for any $v_j\in V'$, we have $k\ge gx$.

For each $v_j\in V'$, we define $A_j$ as the set of agents that meet at $v_j$ and $T_j$ as the total number of moves of agents in $A_j$. 
Then, among agents in $A_j$, the $i$-th smallest number of moves to get to $v_j$ is at least $(i-1)n/k$. Hence, we have
\begin{eqnarray*}
T_j &\ge& \sum_{i=1}^{|A_i|}(i-1)\cdot \frac{n}{k}\\
&\ge& \sum_{i=1}^g (i-1)\cdot \frac{n}{k} + (|A_j|-g)\cdot \frac{gn}{k}\\
    &=& \frac{n}{k}\cdot \frac{g(g-1)}{2} + (|A_j|-g)\cdot \frac{gn}{k}.
\end{eqnarray*}
Therefore, the total number of moves is at least
\begin{eqnarray*}
T &=&   \sum_{v_j\in V'}T_j \\
  &\ge& x\cdot \frac{n}{k}\cdot \frac{g(g-1)}{2} + (k-gx)\cdot \frac{gn}{k}\\
  &=& gn-\frac{gnx}{2k}(g+1).
\end{eqnarray*}
Since $k\ge gx$ holds, we have
\[
T\ge \frac{n}{2} (g-1).
\]
Thus, the total number of moves is   at least $\Omega (gn)$.
\end{proof}

\section{The First Model: A Deterministic Algorithm for Distinct Agents}
\label{sec:deterministic}

In this section, we propose a deterministic algorithm to solve the $g$-partial gathering problem for distinct agents (i.e., agents have distinct IDs). 
The basic idea is that agents elect a leader and then the leader instructs other agents which nodes they meet at. However, since $\Omega(n\log k)$ total moves are  required to elect one leader \cite{Mrs.Suzuki},
this approach cannot lead to the $g$-partial gathering in asymptotically optimal total moves (i.e., $O(gn)$). 
To achieve the partial gathering in $O(gn)$ total moves,
we elect multiple agents as leaders by executing the leader agent election partially. 
By this behavior, the number of moves for the election can be bounded by $O(n\log g)$.
In addition, we show that the total number of moves for agents to move to their gathering nodes by leaders' instruction is
$O(gn)$.
Thus,
%the number of moves for the election can be bounded by $O(n\log g)$, and thus,
our algorithm solves the $g$-partial gathering problem in $O(gn)$ total moves.

The algorithm consists of two parts. 
In the first part, multiple  agents are elected as leader agents. 
In the second part, the leader agents instruct the other agents which nodes they meet at, and the other agents move to the nodes by the instruction.

\subsection{The first part: leader election}\label{distinct}
The aim of the first part is to elect leaders that satisfy the following conditions 
called {\em leader election conditions}: 1) At least one agent is elected as a leader, and 2) there exist at least $g-1$ non-leader agents between two leader agents. To attain this goal, we use a traditional leader election algorithm \cite{Perterson}. However, the algorithm in \cite{Perterson} is executed by nodes and the goal is to elect exactly one leader. 
Hence we modify the algorithm to be executed by agents, and then agents elect 
multiple leader agents by executing the algorithm partially.

During the execution of leader election, the states of agents are divided into the following three types:
\begin{itemize}
\item {\em active}: The agent is performing the leader agent election as a candidate of leaders.
\item {\em inactive}: The agent has dropped out from the candidate of leaders.
\item {\em leader}: The agent has been elected as a leader.
\end{itemize}

For an intuitive understanding, we first explain the idea of leader election by assuming that the ring is synchronous and bidirectional. 
Later, the idea is applied to our model, that is, asynchronous unidirectional rings.
The algorithm consists of several phases. 
In each phase, each active agent compares its own ID with IDs of its backward and forward neighboring active agents. 
More concretely, each active agent $a_h$ writes its own ID $id_2$ on the whiteboard of its current node, 
and  moves backward and forward.
Then, $a_h$ observes ID $id_1$ of its backward active agent and $id_3$ of its forward active agent.
After this, $a_h$ decides if it remains active or drops out from  the candidates of leaders.
Concretely, if its own ID $id_2$ is the smallest among the three IDs, $a_h$ remains active (as a candidate of leaders) in the next phase. Otherwise, $a_h$ drops out from the candidate of leaders and becomes inactive. Note that, in each phase, neighboring active agents never remain as candidates of leaders.  
Thus,  at least half active agents become inactive in each phase.
% and the number of inactive agents between two active agents at least  doubles in each phase.
Moreover  from \cite{Perterson}, after executing $j$ phases, there exist at least $2^j-1$ inactive agents between two active agents. 
Thus, after executing $\lceil \log g \rceil$ phases, the following properties are satisfied: 1) At least one agent remains as a candidate of leaders,  and 2) the number of inactive agents between two active agents is at least $g-1$. 
Therefore, all remaining active agents become leaders since they satisfy the leader election conditions. Note that, before executing $\lceil \log g\rceil$ phases, the number of active agents may become one. In this case, the active agent immediately becomes a leader.

In the following, we implement the above algorithm in asynchronous unidirectional rings. 
First, we implement the above algorithm in a unidirectional ring by applying a traditional technique \cite{Perterson}. 
Let us consider the behavior of active agent $a_h$. 
In unidirectional rings, $a_h$ cannot move backward and  cannot observe the ID of its backward active agent. Instead, $a_h$ moves forward until it observes IDs of two active agents. Then, $a_h$ observes IDs of three successive active agents. 
We assume $a_h$ observes $id_1$, $id_2$, $id_3$ in this order. 
Note that $id_1$ is the ID of $a_h$. 
Here this situation is similar to that the active agent with ID $id_2$ observes $id_1$ as its backward active agent and $id_3$ as its forward active agent in bidirectional rings. For this reason, $a_h$ behaves as if it would be an active agent with ID $id_2$ in bidirectional rings. That is, if $id_2$ is the smallest among the three IDs, $a_h$ remains active as a candidate of leaders. Otherwise, $a_h$ drops out from the candidate of leaders and becomes inactive. 
After the phase if $a_h$ remains active as a candidate, it  assigns $id_2$ to its ID and starts the next phase.\footnote{We imitate the way in \cite{Perterson}, but active agent $a_h$ may still use its own ID $id_1$ in the next phase.}

\begin{figure}[t!]
\begin{center}
\centering
\includegraphics[clip,width=130mm]{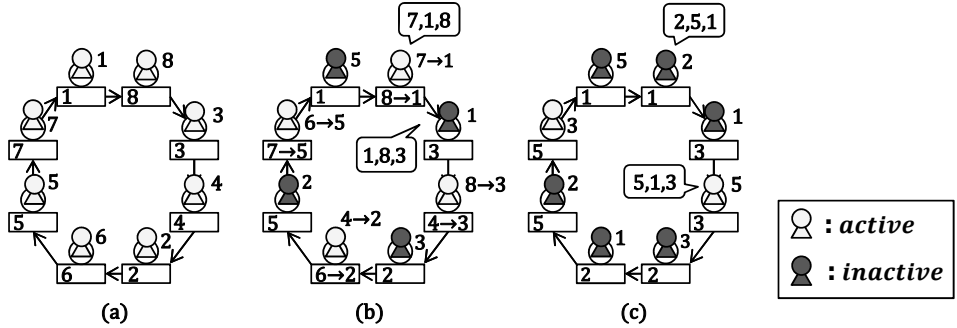}
\caption{An execution example of the leader election part ($k=8,g=3$)}
\label{example}

\end{center}

\end{figure}
For example, consider the initial configuration in Fig.\,\ref{example} (a). 
In the figures, the number near each agent is the ID of the agent and the box of each node represents the whiteboard. 
In the first phase, each agent writes its own ID on the whiteboard of its initial node. Next, each agent moves forward until it observes two IDs, and then the configuration is changed to the one in Fig.\,\ref{example} (b). In this configuration, each agent compares three IDs. 
The agent with ID 1 observes IDs (1, 8, 3), and hence  it drops out from the candidate because the middle ID 8 is not the smallest. 
The agents with IDs 3, 2, and 5 also drop out from the candidates. 
The agent with ID 7 observes IDs (7, 1, 8), and hence  it remains active as a candidate because the middle ID 1 is the smallest. Then, it updates its ID to 1 and starts the next phase. The agents with IDs 8, 4, and 6 also remain active as candidates and similarly update their IDs and start the next phase.
In the second phase, active agents with updated IDs with 1,2,3, and 5 move until they observe two IDs of active agents respectively,
and then the configuration change is changed to the one in Fig.\,\ref{example} (c).
In this configuration, the agent with ID 2 observes IDs (2, 5, 1), and it drops out from the candidate because the middle 
ID is not the smallest.
Similarly, the agent with ID 1 also drops out from the candidate.
On the other hand, the agent with ID 5 observes IDs (5, 1, 3), and it remain active because the middle ID is the smallest.
Similarly, the agent with ID 3 remains active.
Since agents with IDs 5 and 3 execute 2 ($=\lceil \log g\rceil$) phases, they become leaders.

Next, we explain the way to treat asynchronous agents. 
To recognize the current phase, each agent manages {\em a phase number}. 
Initially, the phase number is zero, and it is incremented when each phase is completed. 
Each agent compares IDs with agents that have the same phase number. 
To realize this, when each agent writes its ID on the whiteboard, it also writes its phase number. 
That is, at the beginning of each phase, 
active agent $a_h$ writes a tuple $(phase,id_h)$ on the whiteboard of its current node, 
where $phase$ is the current phase number and $id_h$ is the current ID of $a_h$. 
After that, $a_h$ moves until it observes two IDs with the same phase number as that of $a_h$. 
Note that, some agent $a_h$ may pass another agent $a_i$. In this case, $a_h$ waits until $a_i$ catches  up with $a_h$. We explain the details  later.  Then, $a_h$ decides whether it remains active as a candidate or becomes inactive. If $a_h$ remains active, it updates its own ID. Agents repeat these behaviors until they complete the $\lceil \log g\rceil$-th phase.

{\em Pseudocode.}
The pseudocode to elect leader agents is   given in Algorithm \ref{active} and \ref{basic1}.
All agents start the algorithm with active states, and the behavior of active agent $a_h$ is described in Algorithm \ref{active}. 
We describe $v_j$ by the node that $a_h$ currently exists.
If $a_h$ changes its state to an inactive state or a leader state, $a_h$ immediately moves to the next part and executes the algorithm for an inactive state or a leader state in Section \ref{realization}. 
Agent $a_h$ and node $v_j$ have the following variables:
\begin{itemize}
\item $a_h.id_1,a_h.id_2,$ and $a_h.id_3$ are variables for $a_h$  to store IDs of three successive active agents. 
Agent $a_h$ stores its ID on $a_h.id_1$ and initially assigns its initial ID $a_h.id$ to $a_h.id_1$. 
\item $a_h.phase$ is a variable for  $a_h$ to store its own phase number.
\item $v_j.phase$ and $v_j.id$ are variables for an active agent to write its phase number and its ID. For any $v_j$, initial values of these variables are 0.
\item $v_j.inactive$ is a  variable to represent whether there exists an inactive agent at  $v_j$ or not. That is, agents update the variable to keep the following invariant: If there exists an inactive agent on $v_j$, $v_j.inactive=\textit{true}$ holds, and otherwise $v_j.inactive$$=$$\textit{false}$ holds. Initially $v_j.inactive$ = \textit{false} holds for any $v_j$.
\end{itemize}

In Algorithm \ref{active}, $a_h$ uses procedure  \textit{BasicAction}(), by which agent $a_h$ moves to node $v_{j'}$ satisfying $v_{j'}.phase=a_h.phase$. 
\begin{algorithm}[t!]
\caption{The behavior of active agent $a_h$ ($v_j$ is the current node of $a_h$)}
\label{active}                          
\begin{algorithmic}[1]

\item[\textbf{Variables in Agent $a_h$}]
\item[int $a_h.phase$;] 
\item[int $a_h.id_1$,$a_h.id_2$,$a_h.id_3$;] 

\item[\textbf{Variables in Node $v_j$}]
\item[int $v_j.phase$;] 
\item[int $v_j.id$;]
\item[boolean $v_j.inactive=$ \textit{false};] 

\item[\textbf{Main Routine of Agent $a_h$}] 

\STATE  $a_h.phase=1$ 
\STATE  $a_h.id_1=a_h.id$ 
\STATE  $v_j.phase=a_h.phase$
\STATE  $v_j.id=a_h.id$
\STATE \textit{BasicAction}()
\IF{$(v_j.phase=a_h.phase)\land (v_j.id=a_h.id_1)$}
\STATE change its state to a leader state 
\ENDIF
\STATE  $a_h.id_2=v_j.id$
\STATE \textit{BasicAction}()
\STATE  $a_h.id_3=v_j.id$ 
\IF{$a_h.id_2\geq \min (a_h.id_1,a_h.id_3)$}
\STATE  $v_j.inactive=true$ 
\STATE change its state to an inactive state
\ELSE
\IF {$a_h.phase=\lceil \log g \rceil $}
\STATE change its state to a leader state
\ELSE
\STATE $a_h.phase = a_h.phase+1$
\STATE  $a_h.id_1=a_h.id_2$
\ENDIF
\STATE return to step 3

\ENDIF
\end{algorithmic}
\end{algorithm}
The pseudocode of \textit{BasicAction}() is described in Algorithm \ref{basic1}. In \textit{BasicAction}(), the main behavior of $a_h$ is to move to node $v_{j'}$ satisfying $v_{j'}.phase=a_h.phase$. To realize this, $a_h$ skips nodes where  no agent initially exists (i.e., $v_j.initial=$ \textit{false}) or an inactive agent whose phase number is not equal to $a_h$'s phase number currently exists (i.e., $v_j.inactive=true$ and $a_h.phase \neq v_j.phase$), and continues to move until it reaches a node where some active agent starts the same phase (lines 2 to 4).
Note that during the execution of the algorithm, it is possible that $a_h$ becomes the only one candidate of leaders. In this case, $a_h$ immediately becomes a leader (lines 6 to 8 of Algorithm \ref{active}).

In the following, we explain the details of the treatment of asynchronous agents.
Since agents move asynchronously, agent $a_h$ may pass some active agents. To wait for such agents, agent $a_h$ makes some additional behavior (lines 5 to 8).
First, consider  the transition from the configuration of Fig.\,\ref{pass1} (a) to that of Fig.\,\ref{pass1} (b) and  consider the case that $a_h$ passes $a_b$ with a smaller phase number. Let $x=a_h.phase$ and $y=a_b.phase$ ($y<x$).
In this case, $a_h$ detects the passing when it reaches a node $v_c$ such that $a_h.phase>v_c.phase$ holds.
Hence, $a_h$ can wait for $a_b$ at $v_c$. Since $a_b$
increments $v_c.phase$ or becomes inactive at  $v_c$, $a_h$ waits at $v_c$ until either $v_c.phase=x$ or $v_c.inactive=true$ holds (line 6). 
After $a_b$ updates the value of either $v_c.phase$ or $v_c.inactive$, $a_h$ resumes its behavior.

\if()
\begin{figure}[t!]
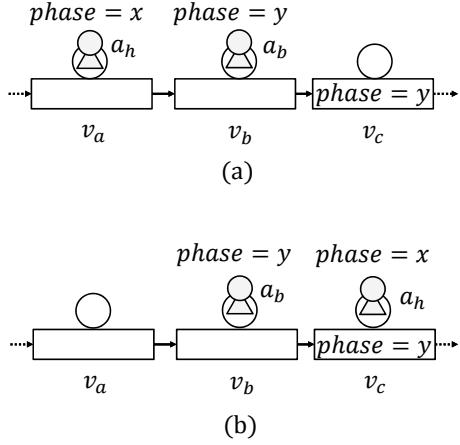

\centering
\subfigure{\includegraphics[width=5cm]{passex1.ps}}
\hspace{1mm}
\subfigure{\includegraphics[width=5cm]{passex2.ps}}
\caption{The first example of agent  $a_h$ that passes other agents (e.g, $a_b$)}
\label{pass1}
\end{figure}
\fi

\begin{algorithm}[t!]
\caption{Procedure \textit{BasicAction}() for $a_h$}         
\label{basic1}                          
\begin{algorithmic}[1]
\STATE move to the forward node
\WHILE {$(v_j.initial=\textit{false})\vee (v_j.inactive=true \land a_h.phase\neq v_j.phase)$}
\STATE move to the forward node
\ENDWHILE
\IF{$a_h.phase > v_j.phase$}
\STATE wait until $v_j.phase=a_h.phase$ or $v_j.inactive=true$
\STATE return to step 2
\ENDIF
\end{algorithmic}
\end{algorithm}

Next, consider the case that $a_h$ passes $a_b$ with the same phase number. In the following, we show that agents can treat this case without any additional procedure. Note that, because $a_h$
increments its phase number after it collects two other IDs, this case happens only when $a_b$ is a forward active agent of $a_h$. Let $x=a_h.phase=a_b.phase$. Let $a_h$, $a_b$, $a_c$, and $a_d$
are successive agents that start phase $x$. Let $v_h$, $v_b$, $v_c$, and $v_d$ are nodes where $a_h$, $a_b$, $a_c$, and $a_d$ start phase $x$, respectively. Note that $a_h$ (resp., $a_b$) decides whether it becomes
inactive or not at $v_c$ (resp., $v_d$). We consider further two cases depending on the decision of $a_h$ at $v_c$. 
First, in  the transition  from the configuration of Fig.\,\ref{pass2} (a) to that of Fig.\,\ref{pass2} (b), consider the case $a_h$ becomes inactive at $v_c$. 
In this case, since $a_h$ does not update $v_c.id$, $a_b$ gets $a_c.id$ at $v_c$ and moves to $v_d$ and then decides its behavior at $v_d$. 
Next, in  the transition  from the configuration of Fig.\,\ref{pass3} (a) to that of Fig.\,\ref{pass3} (b), consider the case $a_h$ remains active at $v_c$. In this case, $a_h$ increments its phase (i.e., $a_h.phase=x+1$) and
updates $v_c.phase$ and $v_c.id$. Note that, since $a_h$ remains active, $a_h.id_2=a_b.id$ is the smallest among the three IDs. Hence, $v_c.id$ is  updated to $a_b.id$ by $a_h$.
%=a_b.id$ holds after that. 
Then, $a_h$ continues to move until
it reaches $v_d$. 
If $a_h$ reaches  $v_d$ before $a_b$ reaches $v_d$, both $v_d.phase<a_h.phase$ and $v_d.inactive=$ \textit{false} hold at $v_d$.
Hence, $a_h$ waits until $a_b$ reaches $v_d$. 
On the other hand when $a_b$ reaches $v_c$,
since $a_b.phase<v_c.phase$ holds, $a_b$ continues to move 
without waiting for the update of $v_c.phase$.
In addition since $a_h$ has updated $v_c.id$,
$a_h$ sees $v_c.id=a_b.id$. 
Thus since  $a_b.id_1=a_b.id_2$ holds, $a_b$ becomes inactive when it reaches $v_d$. After that, $a_h$ resumes the movement.

\begin{figure}[t!]
\begin{center}
\centering
\includegraphics[clip,width=80mm]{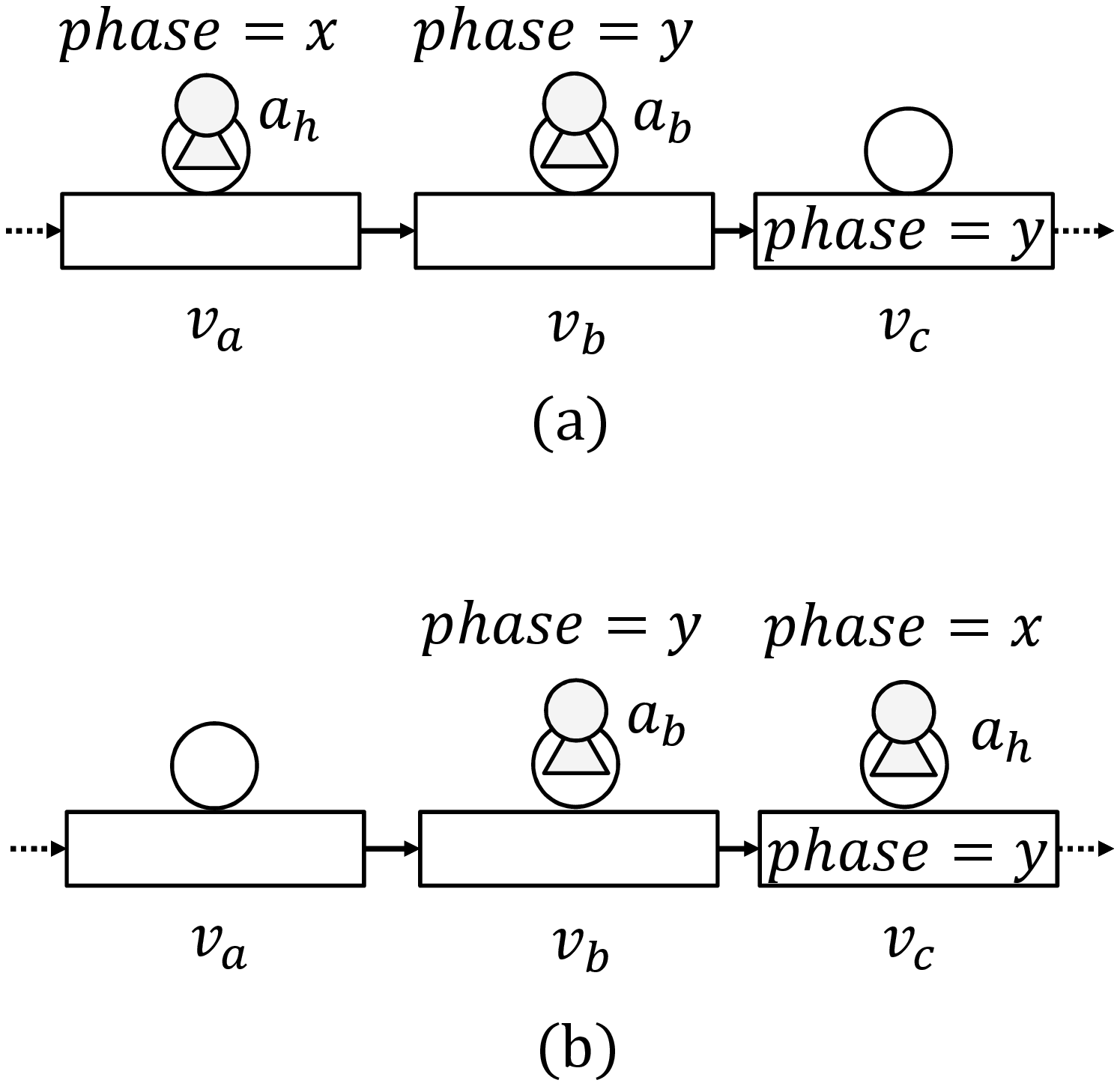}
\caption{The first example of agent  $a_h$ that passes other agents (e.g, $a_b$)}
\label{pass1}
\end{center}
\end{figure}

\begin{figure}[t!]
\begin{center}
\centering
\includegraphics[clip,width=80mm]{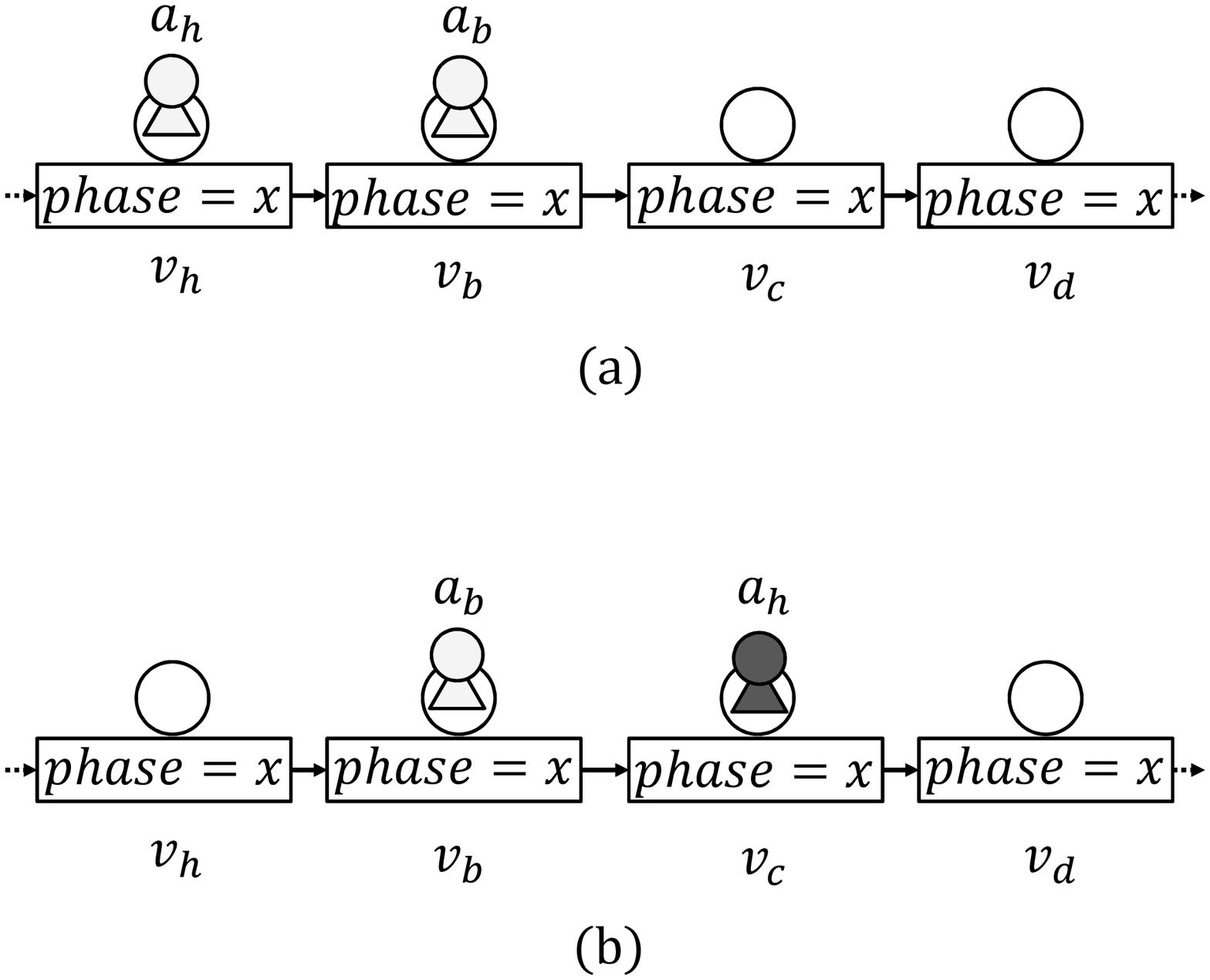}
\caption{The second example of agent $a_h$ that passes other agents (e.g., $a_b$)}
\label{pass2}
\end{center}
\end{figure}

We have the following lemma about Algorithm \ref{active} similarly to \cite{Perterson}.

%We have the following lemma about Algorithm \ref{active} similarly to \cite{Perterson}.
\begin{lemma}
\label{leaders}
Algorithm \ref{active} eventually terminates, and the configuration satisfies the following properties.
\begin{itemize}
\item There exists at least one leader agent.
\item There exist at least $g-1$ inactive agents between two leader agents. 
\end{itemize}
\end{lemma}

\begin{proof}
At first, we show that Algorithm \ref{active} eventually terminates.
After executing $\lceil \log g \rceil$ phases, agents that have dropped out from the candidates of leaders are inactive states, and agents that remain active changes their states to leader states. 
In addition if  agent $a_h$ passes another agent $a_{h'}$, 
$a_h$ waits for $a_{h'}$ at some node $v_j$ until either $v_j.phase$ or $v_j.inactive$ is updated 
(lines 5 to 8 in Algorithm \ref{basic1}).
Since the passed agent $a_{h'}$ eventually reaches $v_j$ and updates  either $v_j.phase $ or $v_j.inactive$,
it does not happen that $a_h$ waits at $v_j$ forever.
Moreover, by the time executing $\lceil \log g\rceil$ phases, if there exists exactly one active agent and the other agents are inactive, the active agent changes its state to a leader state. 
Therefore, Algorithm \ref{active} eventually terminates.
In the following, we show the above two properties.

First, we show that there exists at least one leader agent. From Algorithm \ref{active}, in each phase if $a_h.id_2$ is 
the smallest of the three IDs, $a_h$ remains active. Otherwise, $a_h$ becomes inactive. 
Since each agent uses a unique ID, if there exist at least two active agents in some phase $i$, at least one agent remains active after executing the phase $i$.
Moreover, from lines 6 to 8 of Algorithm \ref{active}, if there exists exactly one candidate of leaders and the other agents remain inactive, the candidate becomes a leader. Therefore, there exists at least one leader agent.

Next, we show that there exist at least $g-1$ inactive agents between two leader agents. At first, we show that after executing $j$ phases, there exist at least $2^j -1$ inactive agents between two active agents.  
We show it by induction on the phase number and by using  the fact that in each phase if   agent $a_h$ remains as a candidate of leaders, then its backward and forward active agents drop out from candidates of leaders.
For the case $j = 1$, there exists at least $1 = 2^1 - 1$ inactive agents between two active agents. For the case $j = l$, we assume that there exist at least $2^l -1$ inactive agents between two active agents. 
Then, after executing $l+1$ phases, since at least one of neighboring active agents becomes inactive, the number of inactive agents between two active agents is at least $(2^l -1) +1 + (2^l - 1)$ = $2^{l+1} - 1$. 
Hence, we can show that after executing $j$ phases, there exist at least $2^j - 1$ inactive agents between two active agents. Therefore, after executing $\lceil \log g\rceil$ phases, there exist at least $g-1$ inactive agents between two leader agents.
\end{proof}
\begin{figure}[t!]
\begin{center}
\centering
\includegraphics[clip,width=80mm]{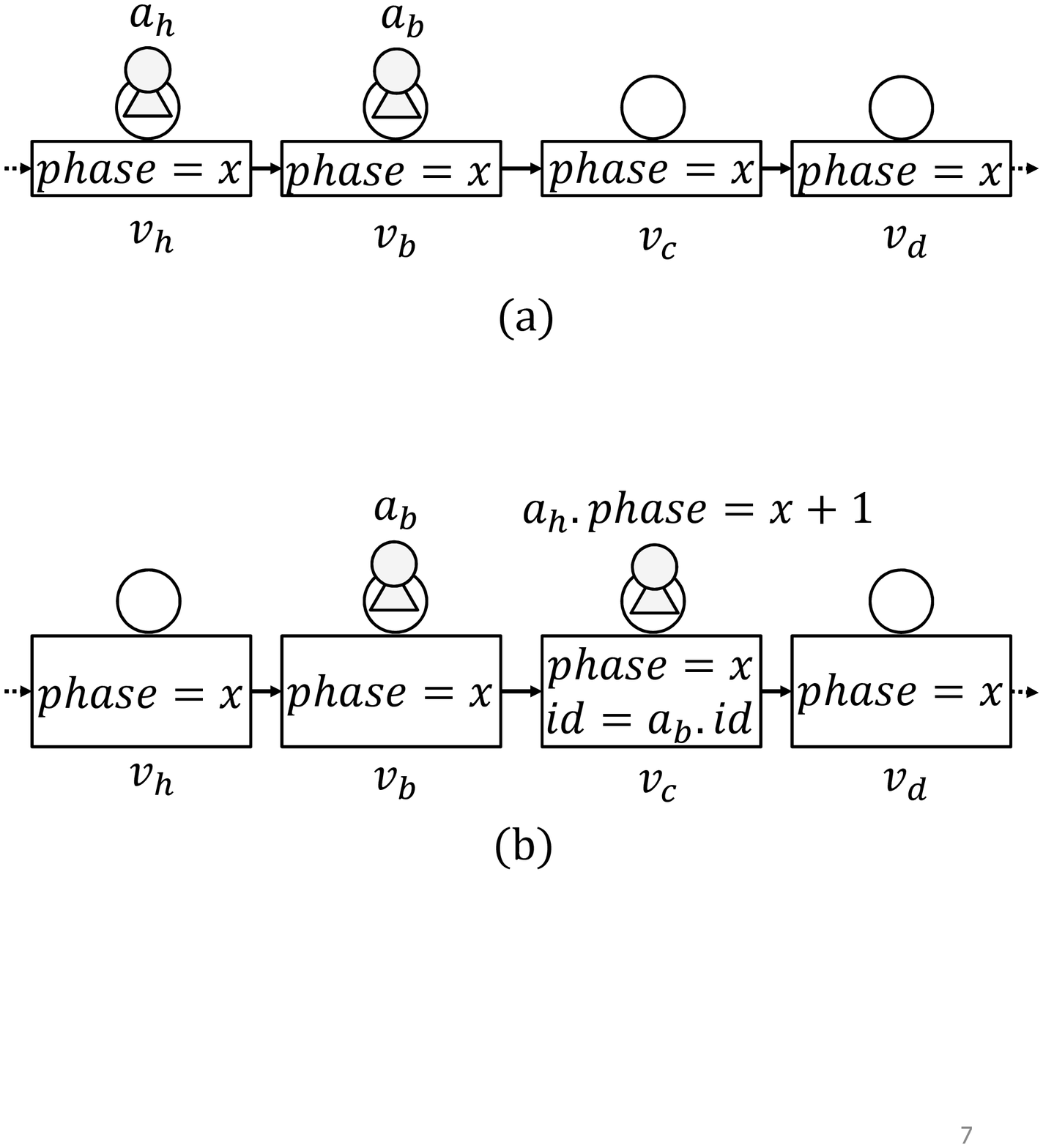}
\caption{The third example of  agent $a_h$ that passes other agents (e.g, $a_b$)}
\label{pass3}
\end{center}
\end{figure}

In addition, we have the following lemma similarly to \cite{Perterson}.

%In addition, we have the following lemma similarly to \cite{Perterson}.
\begin{lemma}
\label{hodai:deterministic}
The total number of agent moves to execute Algorithm \ref{active} is   $O(n\log g)$. 
\end{lemma}

\begin{proof}
In each phase, each active agent moves until it observes two IDs of active agents. 
This costs $O(n)$ moves in total  because each communication link is passed by two agents. 
Since agents execute $\lceil \log g \rceil $ phases, we have the lemma.
\end{proof}

\subsection{The second part: movement to gathering nodes}
\label{realization}

The  second part achieves the $g$-partial gathering by using leaders elected in the first part. 
Let leader nodes (resp., inactive nodes) be the nodes where agents become leaders (resp., inactive agents) in the first part.
In this part, states of agents are divided into the following three types:
\begin{itemize}
\item $leader$: The agent instructs inactive agents where they should move.
\item $inactive$: The agent waits for the leader's instruction.
\item $moving$: The agent moves to its gathering node.
\end{itemize}

 The idea of the algorithm is to divide agents into groups each of which consists of at least $g$ agents.
Concretely, first each leader agent $a_h$ writes 0 on the whiteboard of the current node (i.e., the leader node). 
Next, $a_h$ moves to the next leader node, that is, the node where 0 is already written on the whiteboard.
While moving, whenever $a_h$ visits an inactive node $v_j$, it counts the number of inactive nodes that $a_h$ has
visited. 
If the number plus one is not a multiple of $g$, $a_h$ writes 0 on the whiteboard.
Otherwise, $a_h$ writes 1 on the whiteboard.
These numbers are used to indicate whether the node is a gathering node or not.
%instruct inactive agents where they should move to achieve the $g$-partial gathering. 
The number 0 means that agents do not meet at the node and the number 1 means that at least $g$ agents meet at the node. 
When $a_h$ reaches the next leader node, it changes its own state to a moving state, and we explain the behavior of 
moving agents later. 
For example, consider the configuration in Fig.\,\ref{fig:moving} (a). In this configuration, agents $a_1$ and $a_2$ are leader agents. First, $a_1$ and $a_2$ write 0 on their current whiteboards (Fig.\,\ref{fig:moving} (b)), and then they move and write numbers on whiteboards until they visit the node where 0 is already written on the whiteboard. Then, the system reaches the configuration in Fig.\,\ref{fig:moving} (c).

\begin{figure}[t!]
\begin{center}
\includegraphics[keepaspectratio,width=120mm]{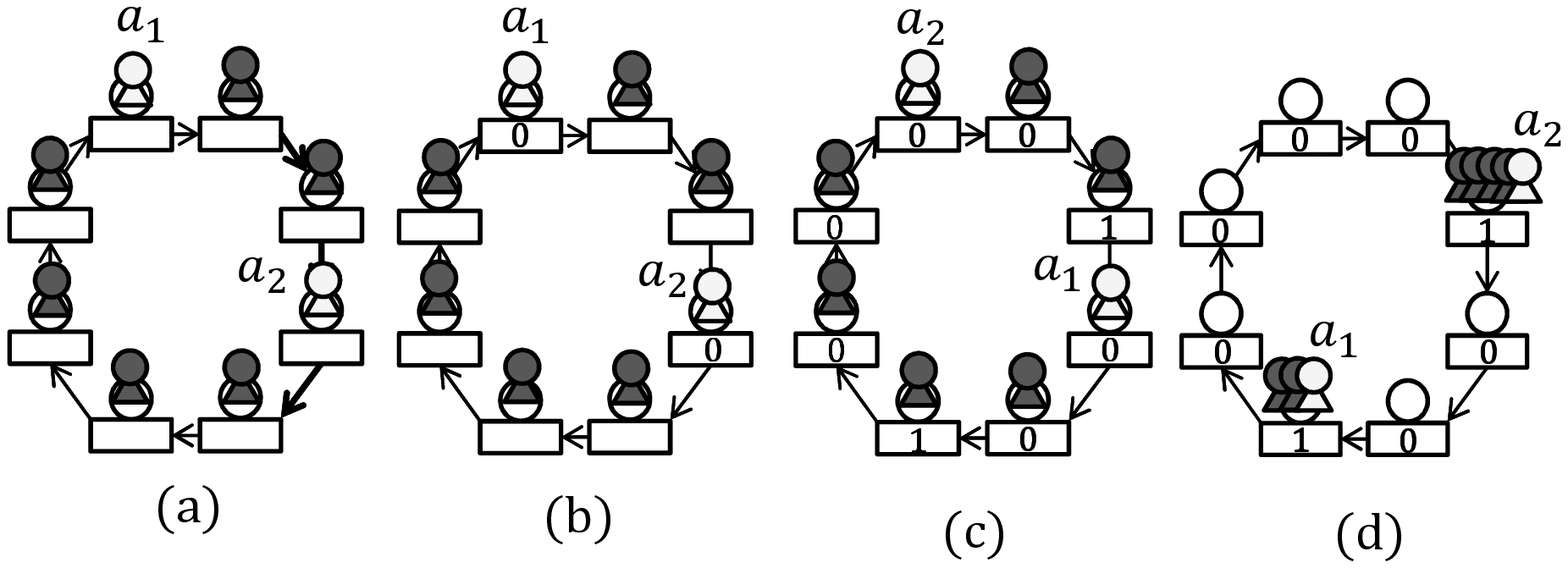}
\caption{The realization of partial gathering ($g=3)$}
\label{fig:moving}
\end{center}
\end{figure}

Each non-leader (i.e., inactive agent) $a_h$ waits at the current node until the value of the whiteboard is updated.
When the value is updated, $a_h$ changes its own state to a moving state.
Each moving agent moves to the nearest node where 1 is written on the whiteboard. 
For example, after the configuration in Fig.\,\ref{fig:moving} (c), each non-leader agent moves to the node where 1 is written on the whiteboard and the system reaches the configuration in Fig.\,\ref{fig:moving} (d). After that, agents can solve the $g$-partial gathering problem.

{\em Pseudocode.}
The pseudocode to achieve the partial gathering is described in Algorithm  \ref{initial} to \ref{moving}.
In this part, agents continue to use $v_j.initial$ and $v_j.inactive$. Remind that $v_j.initial=true$ if and only if there exists an agent at $v_j$ initially. 
In addition,  $v_j.inactive=true$ if and only if there exists an inactive agent at $v_j$. Note that, since each agent becomes inactive or a leader at a node such that there exists an agent initially, agents can ignore and skip every node $v_{j'}$ such that $v_{j'}.initial=\textit{false}$ holds.

At first, the variables needed to achieve the $g$-partial gathering are described in Algorithm \ref{initial}. 
For leader agents instructing inactive agents gathering nodes, agent $a_h$ and node $v_j$ have the following variables:
\begin{itemize}
\item $a_h.count$ is a variable for $a_h$ to count the number of inactive nodes $a_h$ visits 
(The counting is done modulo $g$).
The initial value of $a_h.count$ is 0.
\item $v_j.\textit{isGather}$ is a variable for leader agents to write values to indicate 
whether node $v_j$ is a gathering node or not. 
That is, when a leader agent $a_h$ visits an inactive node $v_j$, $a_h$ writes 1 to $v_j.\textit{isGather}$ to indicate $v_j$ is a gathering node if $a_h.count=0$, and $a_h$ writes 0 to $v_j.\textit{isGather}$ otherwise. 
The initial value of $v_j.\textit{isGather}$ is $\perp$.
\end{itemize}
\begin{algorithm}[t!]
\caption{Initial values needed in the second part ($v_j$ is the current node of agent $a_h$)}
\label{initial}                          
\begin{algorithmic}[1]
\item [\textbf{Variable in Agent $a_h$}]
\item [int $a_h.count=0$;]
\item [\textbf{Variable in Node $v_j$}]
\item [int $v_j.\textit{isGather}=\perp$;]        
\end{algorithmic}
\end{algorithm}

The pseudocode of leader agents is described in Algorithm \ref{leader}. 
Since agents move asynchronously, it is possible that 
there  exists active agents executing the first part and leader agents executing the second part at the same time.
Hence, it may happen that some leader agent $a_h$ may pass some active agent $a_i$.
In this case, $a_h$ waits until $a_i$ catch up with $a_h$ and $a_i$ becomes a leader or inactive. 
More precisely, when leader agent $a_h$ visits the node $v_j$ such that $v_j.initial=true$ and 
$v_j.inactive=$ \textit{false} and $v_j.\textit{isGather}  = \perp$ hold,
it detects that it passes some active agent $a_i$.
This is because $v_j.inactive=true$ should hold if some agent becomes inactive at $v_j$, and $v_j.\textit{isGather}\neq \perp$ holds if some agent becomes leader at $v_j$. 
In this case, $a_h$ waits there until the agent caches up with it and either $v_j.inactive=true$ or $v_j.\textit{isGather}\neq \perp$ holds (lines 8 to 10). When the leader agent updates $v_j.\textit{isGather}$, an inactive agent on node $v_j$ changes to  a moving state (line 17). 
After a leader agent reaches the next leader node, it changes its own state to a moving state (line 22). 
The behavior of inactive agents is described in Algorithm \ref{inactive}.

\begin{algorithm}[t!]
\caption{The behavior of leader agent $a_h$ ($v_j$ is the current node of $a_h$)}
\label{leader}                          
\begin{algorithmic}[1]                  
\STATE $v_j.\textit{isGather}=0$ 
\STATE $a_h.count=a_h.count+1$
\STATE move to the forward node
\WHILE{$v_j.\textit{isGather}=\perp$}
\WHILE {$v_j.initial=$ \textit{false}}
\STATE move to the forward node
\ENDWHILE
\IF{($v_j.inactive=\textit{false})\wedge( v_j.\textit{isGather}=\perp)$}
\STATE wait until $v_j.inactive=true$ or $v_j.\textit{isGather}\neq \perp$
\ENDIF
\IF {$v_j.inactive=true$}
\IF {$a_h.count=0$}
\STATE  $v_j.\textit{isGather}=1$
\ELSE
\STATE  $v_j.\textit{isGather}=0$
\ENDIF
\STATE // an inactive agent at $v_j$ changes to  a moving state
\STATE  $a_h.count=(a_h.count+1)\bmod g$
\STATE move to the forward node
\ENDIF
\ENDWHILE
\STATE change to  a moving state
\end{algorithmic}
\end{algorithm}
\begin{algorithm}[t!]
\caption{The behavior of inactive agent $a_h$ ($v_j$ is the current node of $a_h$)}
\label{inactive}                          
\begin{algorithmic}[1]                  
\STATE wait until $v_j.\textit{isGather}\neq\perp$
\STATE change to  a moving state
\end{algorithmic}
\end{algorithm}

The pseudocode of moving agents is described in Algorithm \ref{moving}.
Moving agent $a_h$ moves to the nearest node $v_j$ such that $v_j.\textit{isGather}=1$ holds. 
When  all agents complete such moves, the $g$-partial gathering problem is solved. 
In asynchronous rings, a moving agent may pass leader agents. 
To avoid this, the moving agent waits until the leader agent catches up with it. 
More precisely, if moving agent $a_h$ visits node $v_j$ such that $v_j.initial=true$ and $v_j.\textit{isGather}  = \perp$ hold, $a_h$ detects that it passed a leader agent. 
Then,   $a_h$ waits there until the leader agent comes and updates $v_j.\textit{isGather}$ (lines 3 to 5).

\begin{algorithm}[t!]
\caption{The behavior of moving agent $a_h$ ($v_j$ is the current node of $a_h$)}
\label{moving}                          
\begin{algorithmic}[1]
\WHILE{$v_j.\textit{isGather}\neq 1$}
\STATE move to the forward node
\IF{$(v_j.initial=\textit{true})\wedge (v_j.\textit{isGather} = \perp)$}
\STATE{wait until $v_j.\textit{isGather}\neq \perp$}
\ENDIF
\ENDWHILE
\end{algorithmic}
\end{algorithm}

We have the following lemma about the algorithm in Section \ref{realization}.
\begin{lemma}
\label{hodai:realization}
After the leader agent election, agents solve the $g$-partial gathering problem in $O(gn)$ total moves.
\end{lemma}

\begin{proof}
At first, we show the correctness of the proposed algorithm. 
Let $v^g_0,v^g_1,\ldots ,v^g_l$ be nodes such that $v^g_j.\textit{isGather} = 1$ holds ($0\le j\leq l$)
after all leader agents complete their  behaviors, and we call these nodes {\em gathering nodes}. 
From Algorithm \ref{moving}, each moving agent moves to the nearest gathering node $v^g_j$. 
By Lemma \ref{leaders}, 
there exist at least $g-1$ moving agents between $v^g_j$ and $v^g_{j+1}$ 
Hence, agents can solve the $g$-partial gathering problem. 
In the following, we consider the total number of moves required to execute the algorithm.

First, the  total number of moves required for each leader agent to move to its next leader node is obviously $n$. 
Next, let us consider the total number of moves required 
for each moving agent to move to nearest gathering node $v^g_j$ 
 (For example, the total moves from Fig \ref{fig:moving} (c) to Fig \ref{fig:moving} (d)). 
Remind that there are at least $g-1$ inactive agents between two leader agents 
and each leader agent $a_h$ writes 1 to $v_j.\textit{isGather}$ after writing 0 $g-1$ times. 
Hence, there are at most $2g-1$ moving agents between $v^g_j$ and $v^g_{j+1}$.
Thus, the total number of these moves is  $O(gn)$ because each link is passed by at most $2g$ agents. Therefore, we have the lemma.
\end{proof}

From Lemmas \ref{hodai:deterministic} and \ref{hodai:realization}, we have the following theorem. 
\begin{theorem}
\label{teiri:deterministic}
When agents have distinct IDs, our deterministic algorithm solves the $g$-partial gathering problem in $O(gn)$ total moves. \qed

\end{theorem}

\section{The Second Model: A Randomized Algorithm for Anonymous Agents}
\label{sec:randomized}

In this section, we propose a randomized algorithm to solve the $g$-partial gathering problem for anonymous agents under the assumption that each agent knows the total number $k$ of agents.  The idea of the algorithm is the same as that in Section \ref{sec:deterministic}.
In the first part, agents execute the leader election partially and elect multiple leader agents.
In the second part, the leader agents determine gathering nodes and all agents move to the nearest gathering nodes.
In the previous section each agent uses distinct IDs to elect multiple leader agents, 
but in this section each agent is anonymous and uses random IDs.
We also show that the $g$-partial gathering problem is solved in  $O(gn)$ expected total moves.

\subsection{The first part: leader election}
\label{anonymous}

In this subsection, we explain a randomized algorithm to elect multiple leaders by using random IDs.
Similarly to Section \ref{distinct}, the aim in this part is to satisfy the following 
conditions (leader election conditions): 1) At least one agent is elected as a leader, and 2) there exist at least $g-1$ non-leader agents between two leader agents.
The basic idea is the same as   Section \ref{distinct}, that is, each active agent moves in the ring and compares three random IDs.
If the ID in the middle is the smallest of the three random IDs, the active agent remains active.
Otherwise, the active agent drops out from the candidate of leaders.

Now we explain  details of the algorithm.
In the beginning of each phase, each active agent selects $3\log k$ random bits  as its own ID.
After this, each agent executes in the same way as Section \ref{distinct}, that is,
each active agent moves until it observes two random IDs of active agents and compares three random IDs.
If the observed three IDs are distinct,
the agent can execute the leader agent election similarly to Section \ref{distinct}.
In addition to the behavior of the leader election in Section 3.1,
when an agent becomes a leader at  node $v_j$, the agent sets a $\textit{leader}$-$\textit{flag}$ at $v_j$,
and we explain how leader-flags are used later.
%This flag is used to notify semi-leader agents that a leader agent exists.
If no agent observes a same random ID, 
the total number of moves for the leader agent election is the same as in Section \ref{distinct}, that is, $O(n\log g)$.
In the following, we consider the case that some agent observes a same random ID.

Let $a_h.id_1,a_h.id_2$, and $a_h.id_3$ be random IDs that an active agent $a_h$ observes in some phase.
If $a_h.id_1=a_h.id_3\neq a_h.id_2$ holds, then $a_h$ behaves similarly to Section \ref{distinct}, that is,
if $a_h.id_2<a_h.id_1=a_h.id_3$ holds,  $a_h$ remains active and $a_h$ becomes inactive otherwise.
For example, let us consider a  configuration of  Fig.\,\ref{fig:random1} (a).
\begin{figure}[t!]
\centering
\includegraphics[keepaspectratio,width=85mm]{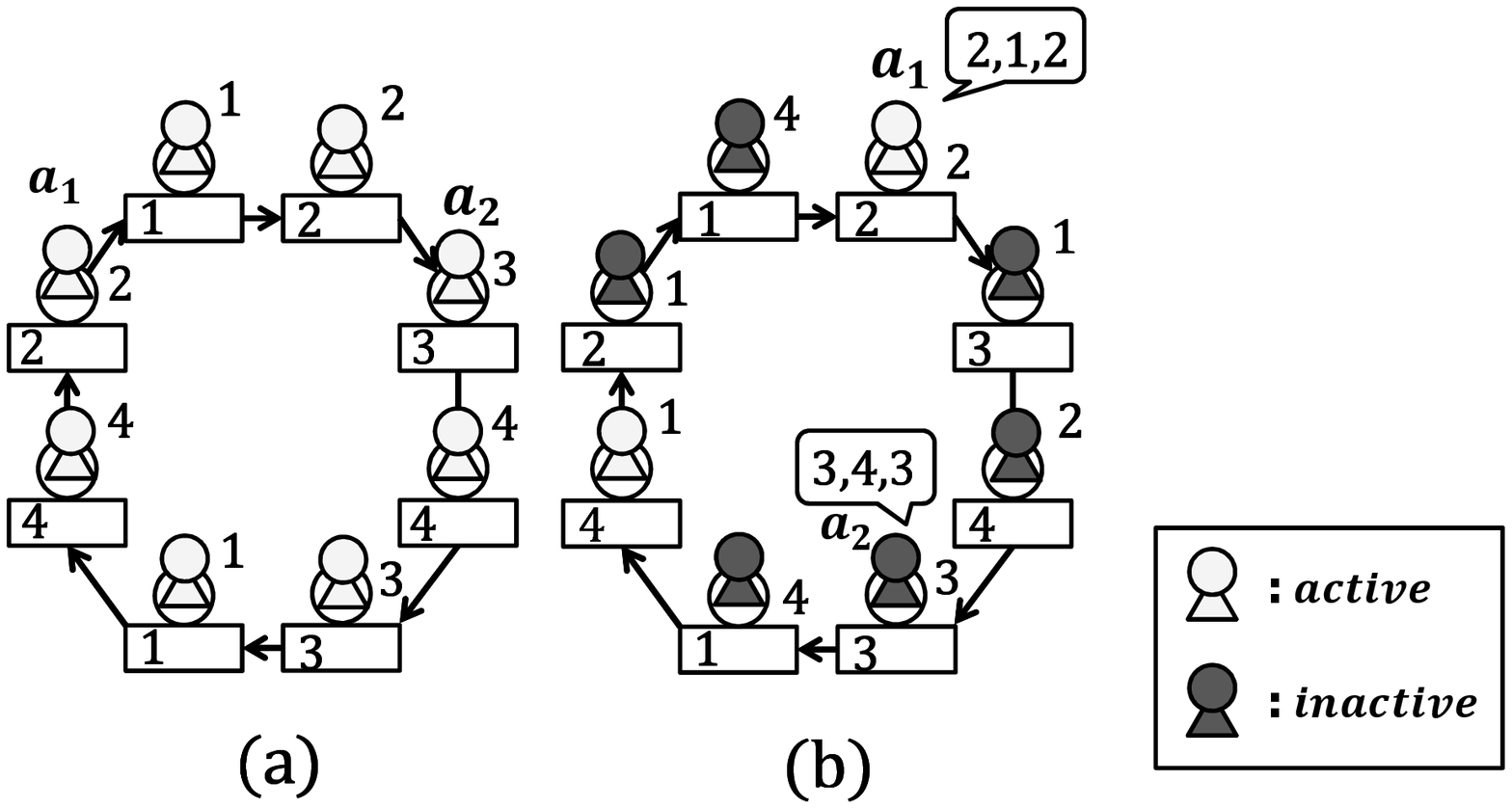}
\caption{An example that some agent observes the same random IDs}
\label{fig:random1}
\end{figure}
Each active agent  moves until it observes two random IDs (Fig.\,\ref{fig:random1} (b)).
Then, agent $a_1$ observes  three random IDs (2,1,2) and remains active 
because $a_1.id_2<a_1.id_1=a_1.id_3$ holds.
On the other hand,  agent $a_2$ observes three random IDs (3,4,3) and becomes inactive because 
$a_2.id_2>a_2.id_1=a_2.id_3$ holds.
The other agents do not observe  same random IDs and behave similarly to Section \ref{distinct}, that is,
if their middle IDs are the smallest, they remain active and execute the next phase.
If their middle IDs are not the smallest, they become inactive.

Next, we consider the case that either $a_h.id_2=a_h.id_1$ or $a_h.id_2=a_h.id_3$ holds.
In this case, $a_h$ changes its own state to a \textit{semi-leader} state.
A semi-leader is an agent that has a possibility to  become a leader if there exists no leader agent in the ring.
When at least one  agent becomes a semi-leader, each active agent becomes inactive.
The outline  of the behavior of each semi-leader agent is as follows:
First each semi-leader travels once around  in the ring. 
After this, if there already exists a leader agent in the ring, each semi-leader becomes inactive.
Otherwise,
the leader election is executed among all semi-leader agents, and 
exactly one semi-leader is elected as a leader and the other agents become inactive (including active agents).
Note that, we can show that the probability some active agent becomes a semi-leader is sufficiently low
and the expected number of semi-leader agents during the leader election is also sufficiently small.
Hence even when each semi-leader travels once around  the ring several times,
the expected total moves to complete the leader agent election can be bounded by $O(n\log g)$.

Now, we explain the detailed behavior for semi-leader agents.
When an  active agent $a_h$ becomes a semi-leader, 
it  sets a \textit{semi-leader-flag} on its current whiteboard.
In the following, the node where the semi-leader flag is set (resp., not set) is called \textit{a  semi-leader node} (resp., \textit{a non-semi-leader node}).
After that, semi-leader agent $a_h$ travels once around  the ring.
In the travel, 
when $a_h$ visits a non-semi-leader node $v_j$ where there exists an agent in the initial configuration, that is, 
a non-semi-leader node $v_j$ such that $v_j.initial=true$ holds,
$a_h$ sets the \textit{tour-flag} at $v_j$.
This flag is used so that other agents notice the existence of a semi-leader and become inactive.
Moreover
when  $a_h$ visits a semi-leader node, $a_h$ 
compares its random ID with the random ID written on the current whiteboard.
Then, $a_h$ memorizes whether its random ID is smaller or not and whether another semi-leader has the same random ID as its  random ID or not.

After traveling once  around  the ring, 
$a_h$ decides if it becomes a leader or inactive.
While traveling in the ring, if $a_h$ observes a leader-flag, 
it learns that there already exists  a leader agent in the ring. 
In this case, $a_h$ becomes inactive.
Otherwise, $a_h$ decides if it  becomes a leader or inactive depending on random IDs.
Let $a_h.id$ be $a_h$'s random ID and $A_{min}$ be the set of semi-leaders such that 
each semi-leader $a_h\in A_{min}$ has the smallest random ID $id_{min}$ among all semi-leaders.
In this case, each semi-leader $a_h\notin A_{min}$ clears a semi-leaders-flag and becomes inactive.
On the other hand, if $a_h$ has the unique minimum random ID (i.e., $|A_{min}|=1$),
$a_h$ becomes a leader.
Otherwise, $a_h$ selects a random ID again, writes the ID  on the current whiteboard,
travels once around   the ring.
Then, $a_h$ obtains new  random IDs of semi-leaders.
Each semi-leader $a_h$ repeats such a behavior until $|A_{min}| = 1$ holds.

\if()
For example, let us consider Fig.\,\ref{fig:semi-leaders}.
For simplicity,  we omit nodes with non-leaders.
Each number in the whiteboard represents a random ID of the semi-leader at the current node.
First in Fig.\,\ref{fig:semi-leaders} (a), 
$a_1$'s random ID 1 is the smallest and there exists no semi-leader whose random ID is 1.
Hence, $a_1$ becomes a leader and the other semi-leaders become inactive after they travel once  around  the ring.
Next in Fig.\,\ref{fig:semi-leaders} (b),
there exist several semi-leaders $a_1,a_2$, and $a_4$ whose random IDs are the smallest.
In this case, agents $a_3,a_5$, and $a_6$ drop out from candidates because their random IDs are not the smallest.
On the other hand,
$a_1,a_2$, and $a_4$ select a random ID again, write the ID  on the current whiteboard, 
and circulate the ring respectively.
After this, we assume that  configuration becomes the one of   Fig.\,\ref{fig:semi-leaders} (c).
Then, $a_1$ becomes a leader since its random ID 1 is the smallest and the other random IDs are not 1.
On the other hand, $a_2$ and $a_4$ become inactive.

\begin{figure}[t!]
\centering
\includegraphics[keepaspectratio,width=90mm]{tmpLeaders.ps}
\caption{The behavior of semi-leaders}
\label{fig:semi-leaders}
\end{figure}
\fi

{\em Pseudocode.} The pseudocode to elect leader agents  is given  in Algorithm \ref{algo:randomVariables} to \ref{algo:semiJudge}.
Algorithm \ref{algo:randomVariables} represents variables required for the behavior of active  agents,
and Algorithm \ref{algo:randomActive} represents the behavior of active agents.
Agent $a_h$ and node $v_j$ have the following variables:
\begin{itemize}

\item $a_h.id_1,a_h.id_2,$ and $a_h.id_3$ are variables 
for $a_h$ to store  random IDs of three successive active agents.
Note that $a_h$ stores its own random ID on $a_h.id_1$.

\item $a_h.phase$ is a variable for $a_h$ to store its phase number.

\item $v_j.phase$ and $v_j.id$ are variables for an active agent to write its phase number and its random ID. 
 For every $v_j$, initial values of these variables are 0.

\item $v_j.\textit{tour}$-$\textit{flag}$ and $v_j.leader$-$\textit{flag}$ are variables to  represent whether there exists an semi-leader agent and a leader agent or not respectively. The initial values of these variables are $false$.

\item $a_h.\textit{semiObserve}$ is a variable for $a_h$ to decide whether it  observes a tour-flag or not.
The initial value of $a_h.\textit{semiObserve}$ is $false$.
\end{itemize}
In addition to these variables, agents $a_h$ uses the procedure \textit{random}($l$) to get its own random ID.
This procedure returns $l$ random bits. 	

\begin{algorithm}[t!]
\caption{Values required for the behavior of active agent $a_h$ ($v_j$ is the current node of $a_h$)}
\label{algo:randomVariables}
\begin{algorithmic}[1]

\item[\textbf{Variables for Agent $a_h$}]
\item[int $a_h.phase$;] 
\item[int $a_h.id_1$,$a_h.id_2$,$a_h.id_3$;] 
\item[boolean $a_h.\textit{semiObserve}=$ \textit{false}]
\item[\textbf{Variables for Node $v_j$}]
\item[int $v_j.phase$;] 
\item[int $v_j.id$;]
\item[boolean $v_j.inactive=$ \textit{false};] 
\item[boolean $v_j.\textit{tour}$-$\textit{flag}=$ \textit{false};] 
\item[boolean $v_j.leader$-$\textit{flag}=$ \textit{false};] 
\end{algorithmic}
\end{algorithm}

\begin{algorithm}[t!]
\caption{The behavior of active agent $a_h$ ($v_j$ is the current node of $a_h$)}
\label{algo:randomActive}
\begin{algorithmic}[1]

\STATE  $a_h.phase=1$
\STATE  $a_h.id_1=$ \textit{random}$(3\log k)$
\STATE $v_j.phase=a_h.phase$
\STATE $v_j.id=a_h.id_1$
\STATE \textit{BasicAction}()
\IF {$v_j.\textit{tour}=\textit{true}$}
\STATE $a_h.\textit{\textit{semiObserve}}=\textit{true}$
\ENDIF
\STATE  $a_h.id_2=v_j.id$
\STATE \textit{BasicAction}()
\IF {$v_j.\textit{tour}=\textit{true}$}
\STATE $a_h.\textit{semiObserve}=\textit{true}$
\ENDIF
\STATE  $a_h.id_3=v_j.id$ 
\IF {$a_h.\textit{semiObserve}=\textit{true}$}
\STATE $v_j.inactive = \textit{true}$
\STATE change its state to an inactive state
\ENDIF
\IF {($a_h.phase=v_j.phase)\land (a_h.id_1=a_h.id_2 \lor a_h.id_2=a_h.id_3)$}
\STATE change its state to a semi-leader state 
\ENDIF
\IF{$a_h.id_2\geq \min (a_h.id_1,a_h.id_3)$}
\STATE  $v_j.inactive=\textit{true}$ 
\STATE change its state to an  inactive state 
\ELSE
\IF {$a_h.phase=\lceil \log g \rceil $}
\STATE $v_j.\textit{leader}$-$\textit{flag}=\textit{true}$
\STATE change its state to a leader state
\ELSE
\STATE $a_h.phase = a_h.phase+1$
\ENDIF
\STATE return to step 2
\ENDIF
\end{algorithmic}
\end{algorithm}

In each phase,
each active agent selects  its own random ID of $3\log k$ bits length through \textit{random}($3\log k$),
and moves until it observes two random IDs by \textit{BasicAction}() in Algorithm \ref{basic1}.
If each active agent $a_h$ neither observes a tour-flag nor observes phase numbers and random IDs such that $(a_h.phase=v_j.phase )\land (a_h.id_2=a_h.id_1 \lor a_h.id_2=a_h.id_3)$ holds, 
this pseudocode works similarly to Algorithm \ref{distinct}.
In this case when an agent becomes a leader, the agent sets a leader-flag at $v_j$ (lines 26 to 29). 
If an active agent $a_h$ observes a tour-flag, 
then $a_h$ moves until it observes two random IDs of active agents and becomes inactive (lines 15 to 18).
Remind that $v_j.inactive$ is a variable to represent whether there exists an inactive agent or not.
If an active agent $a_h$ observes three random IDs 
such that $(a_h.phase=v_j.phase )\land (a_h.id_2=a_h.id_1 \lor a_h.id_2=a_h.id_3)$ holds,
then $a_h$ changes its own state to a semi-leader state (lines 19 to 21).
\begin{algorithm}[t!]
\caption{Values required for the behavior of semi-leader agent $a_h$ ($v_j$ is the current node of $a_h$)}
\label{algo:semiValues}                          
\begin{algorithmic}[1]
\item[\textbf{Variables for Agent $a_h$}]
\item[int] $a_h.\textit{semiPhase}$;
\item[int] $a_h.\textit{semiID}$;
\item[int] $a_h.\textit{agentCount}$;
%\item[int] $a_h.N_{tour}$;
\item[boolean] $a_h.\textit{isMin}=$ \textit{true}
\item[boolean] $a_h.\textit{isUnique}=$ \textit{true}
\item[boolean] $a_h.\textit{leaderObserve}=$ \textit{false}

\item[\textbf{Variables for Node $v_j$}]
\item[int] $v_j,\textit{semiPhase}$;
\item[int] $v_j.id$;
\item[boolean] $v_j.\textit{leader}$-$\textit{flag}$; 
\item[boolean] $v_j.\textit{semi}$-$\textit{leader}$-$\textit{flag}$;
\item[boolean] $v_j.\textit{tour}$-$\textit{flag}$;
\end{algorithmic}
\end{algorithm}

\begin{algorithm}[h!]
\caption{The first half behavior of semi-leader agent $a_h$ ($v_j$ is the current node of $a_h$)}
\label{algo:semiRing}                          
\begin{algorithmic}[1]
\IF {$v_j.tour$-$\textit{flag}=\textit{true}$}
\STATE $v_j.inactive =\textit{true}$
\STATE change its state to an inactive state 
\ENDIF
\STATE $v_j.semi$-$leader$-$\textit{flag}=\textit{true}$
\STATE $a_h.\textit{semiPhase}=1$
\STATE $v_j.\textit{semiPhase}=a_h.\textit{semiPhase}$
\STATE $v_j.id = random(3\log k)$
\STATE $a_h.\textit{semiID}= v_j.id$
%\STATE $a_h.\textit{nMin}=0$
%\STATE $a_h.\textit{semiIDs}[a_h.x]=v_j.id$
\WHILE {$a_h.\textit{agentCount}\neq k$}
\STATE {\small move to the forward node}
\WHILE {$v_j.initial=$ \textit{false}}
\STATE move to the forward node 
\ENDWHILE
\STATE $a_h.\textit{agentCount}=a_h.\textit{agentCount}+1$
\IF {$v_j.leader$-$\textit{flag}=\textit{true}$}
\STATE $a_h.\textit{leaderObserve}=\textit{true}$
\ENDIF

\IF {$v_j.semi$-$leader$-$\textit{flag}=\textit{true}$}
\IF {$a_h.\textit{semiPhase}\neq v_j.\textit{semiPhase}$}
\STATE wait until $a_h.\textit{semiPhase}=v_j.\textit{semiPhase}$
\ENDIF
\IF {${v_j.id}<a_h.\textit{semiID}$}
\STATE $a_h.\textit{isMin}=\textit{false}$
\ENDIF
\IF {${v_j.id}=a_h.\textit{semiID}$}
\STATE $a_h.\textit{isUnique}=\textit{false}$
\ENDIF

\ELSE
%\IF {$v_j.tour=$ \textit{false}}
\STATE $v_j.\textit{tour}$-$\textit{flag}=\textit{true}$
\ENDIF
\ENDWHILE

\end{algorithmic}
\end{algorithm}

\begin{algorithm}[h!]
\caption{The latter half behavior of semi-leader agent $a_h$ ($v_j$ is the current node of $a_h$)}
\label{algo:semiJudge}                          
\begin{algorithmic}[1] 
\IF{$a_h.\textit{leaderObserve}=\textit{true}$}
\STATE $v_j.inactive =\textit{true}$
\STATE change its state to an inactive state  
\ENDIF 
\IF {$a_h.\textit{isMin}= \textit{false}$}
\STATE $v_j.\textit{semi}$-$leader$-$\textit{flag}=$ \textit{false}
\STATE $v_j.inactive = \textit{true}$
\STATE change its state to an inactive state 
\ENDIF
\IF {$a_h.\textit{isUnique}=\textit{true}$}
\STATE change its state to a  leader state
\ELSE
\STATE $a_h.\textit{semiPhase}=a_h.\textit{semiPhase}+1$
\STATE $a_h.\textit{agentCount}=0$
%\STATE $v_j.id=$ \textit{random}$(3\log k)$
\STATE return to step 7 of Algorithm \ref{algo:semiRing}
\ENDIF

\end{algorithmic}
\end{algorithm}

Algorithm \ref{algo:semiValues} represents variables required for the behavior of  semi-leader agents,
and Algorithm \ref{algo:semiRing} and Algorithm \ref{algo:semiJudge} represent the behavior of semi-leader agents.
Semi-leader-agent  $a_h$ and node $v_j$ have the following variables:
\begin{itemize}
\item $a_h.\textit{semiID}$ is a variable for $a_h$ to store its random ID.
\item $a_h.\textit{agentCount}$ is a variable for $a_h$ to detect the  completion of one round of the ring travel.
\item $a_h.\textit{isMin}$ is a variable for $a_h$ to detect whether its random ID is the smallest or not.
The initial value of $a_h.\textit{isMin}$ is $true$.
\item $a_h$.$\textit{isUnique}$ is a variable for $a_h$ to detect whether another semi-leader has the same random ID as its  random ID.
The initial value of $a_h$.\\$\textit{isUnique}$ is $true$.
\item $a_h.\textit{leaderObserve}$ is a variable for $a_h$ to detect 
whether there exists a leader agent in the ring or not.
The initial value of $a_h.\textit{leaderObserve}$ is false.
\item $a_h.\textit{semiPhase}$ is a variable for $a_h$ to store its phase number in the semi-leader state.
\item $v_j.\textit{semiPhase}$ is a variable for a semi-leader agent to write its phase number in the semi-leader state.
\end{itemize}
Variables $a_h.\textit{semiPhase}$ and $v_j.\textit{semiPhase}$ are used for the case that there exist several semi-leaders 
having the same smallest random IDs.
In addition to these variables, each node $v_j$ has variables $v_j.id$, $v_j.\textit{leader}$-$\textit{flag}$, $v_j.\textit{semi}$-$\textit{leader}$-$\textit{flag}$, and $v_j.tour$-$\textit{flag}$ as defined in Algorithm \ref{algo:randomVariables}.

Before semi-leader $a_h$ begins moving in the ring (from  $v_j$), 
if it detects tour-flag at $v_j$, 
another semi-leader $a_{h'}$ has already visited $v_j$. 
Then $a_h$ becomes inactive and does not start the travel in the ring (lines 1 to 4 of Algorithm \ref{algo:semiRing}).
This is because, otherwise,
each semi-leader cannot share the same random IDs.
After each semi-leader travels once around   the ring,
if there exists exactly one semi-leader whose random ID is the smallest, the semi-leader becomes a leader
and the other semi-leaders  become inactive.
Otherwise, each semi-leader $a_h$ whose random ID is the smallest 
 updates its phase and random ID again,
and travels once around   the ring (lines 12 to 15 of Algorithm \ref{algo:semiJudge}).
Then, $a_h$ obtains new value of random IDs. 
Each semi-leader $a_h$ repeats  such a behavior until exactly one semi-leader has the smallest random ID.

We have the following lemmas similarly to Section \ref{distinct}.

\begin{lemma}
\label{lem:randomActive}
Algorithm \ref{algo:randomActive} eventually terminates, and the  configuration satisfies the following properties.

\begin{itemize}
\item There exists at least one leader agent.
\item There exist at least $g-1$ inactive agents between two leader agents. 
\end{itemize}

\end{lemma}

\begin{proof}
The above properties are the same as Lemma \ref{active}.
Thus, if no agent becomes  a  semi-leader during the algorithm,
each agent behaves similarly to Section \ref{distinct} and the above properties are satisfied.
Moreover if at least one agent  becomes a semi-leader,
exactly one semi-leader is elected as a leader and the other agents become inactive.
Then, the above properties are clearly satisfied.
Therefore, we have the lemma.
\end{proof}

\begin{lemma}
\label{lem:randomMoves}
The expected total number of agent moves to elect multiple leader agents  is  $O(n\log g)$. 

\end{lemma}

\begin{proof}
If there exist no neighboring active agents having  the same random IDs, 
Algorithms \ref{algo:randomActive} works similarly to Section \ref{distinct},
and the total number of moves is  $O(n\log g)$.
In the following, we consider the case that some neighboring active agents have the same random IDs.

Let $l$ be the length of a random ID.
Then, the probability that two active neighboring active agents have the same random ID is $(\frac{1}{2})^l$.
Thus, when there exist $k_i$ active agents in the $i$-th phase, 
the probability that there exist neighboring active agents having  the same random IDs is at most $k_i\times (\frac{1}{2})^l$.
Since at least half active agents drop out from candidates in each phase, 
the probability that neighboring active agents  have the same random IDs until the end of the $\lceil \log g\rceil$ phases is at most 
$k\times (\frac{1}{2})^l+\frac{k}{2} \times (\frac{1}{2})^l+\cdots +\frac{k}{2^{\lceil \log g\rceil -1}} \times (\frac{1}{2})^l<2k\times (\frac{1}{2})^l$.
Since $l=3\log k$ holds, 
the probability is at most $\frac{2}{k^2}<\frac{1}{k}$.
We assume that $k$ active agents become semi-leaders and circulate  around the ring 
because this case requires the most total moves.
Then, each semi-leader $a_h$ compares its random ID with random IDs of each semi-leader.
Let $A_{min}$ be the set of semi-leader agents whose random IDs are the smallest.
If $|A_{min}|=1$ holds,
agents finish the leader agent election and the total number of moves is  at most $O(kn)$.
Otherwise,
at least two semi-leaders have the same smallest random IDs.
This probability is at most $k\times (\frac{1}{2})^l$.
In this case, each semi-leader $a_h$ updates its phase and random ID again, travels once around   the ring,
and  obtains new random IDs of each semi-leader.
Each semi-leader $a_h$ repeats  such a behavior until $|A_{min}|=1$ holds. 
We assume that $t=k\times (\frac{1}{2})^l $ and semi-leaders complete  the leader agent election after they circulate  around the ring $s$ times.
In this case, before they circulate around the ring $s-1$ times,
$|A_{min}|\neq 1$ holds every time  they circulate around the ring.
In addition when they circulate around the ring $s$ times, $|A_{min}|=1$ holds, and
the probability such that $|A_{min}|=1$ holds is clearly less than 1.
Hence, the probability such that agents complete the leader election after they circulate around the ring $s$ times 
is at most $t^{s-1}\times 1=t^{s-1}$, and 
the total number of  moves is  at most $skn$.
Since the probability that at least one agent becomes a semi-leader is at most $\frac{1}{k}$,
the expected total number of moves for the case that some agents become semi-leaders and complete  the leader agent election
is at most $O(n\log g)+\frac{1}{k}\times \sum^\infty_{s=1}t^{s-1} \times skn = n \sum^\infty_{s=1}st^{s-1}$.
Let $S_n$ be $1\times 1 + 2\times t +\cdots +nt^{n-1}$.
Then, we have $S_n= (nt^{n+1}-(n-1)t^n+1)/(1-t)^2$.
When $n=\infty$, we have $S_n=1/(1-t)^2$.
Moreover  since $t=k\times (\frac{1}{2})^l$ and $l=3\log k$ hold, 
we have $t<\frac{1}{2}$ and $S_n<4$.
Furthermore, the expected total number of moves is  at most $O(n)$.
Since the total moves to elect multiple leaders for the case that no agent becomes a semi-leader is $O(n\log g)$,
the expected total moves for the leader election is $O(n\log g)$.

Therefore, we have the lemma. 
\end{proof}

\subsection{The second part: movement to gathering nodes}
\label{realization2}

After executing the leader agent election in Section \ref{anonymous}, 
the conditions  shown by Lemma \ref{lem:randomActive} is satisfied, that is,
1) At least one agent is elected as a leader,  and 2) there exist at least $g-1$ inactive agents between two leader agents.
Thus, we can execute the algorithms in Section \ref{realization} after the algorithms in Section \ref{anonymous}.
Therefore, agents can solve the $g$-partial gathering problem.

From Lemmas \ref{hodai:realization}, \ref{lem:randomActive}, and \ref{lem:randomMoves}, we have the following theorem.
\begin{theorem}
\label{teiri:randomized}
When agents have no IDs, our randomized algorithm solves the $g$-partial gathering problem in expected $O(gn)$ total moves.\qed
\end{theorem}

\section{The Third Model: A Deterministic Algorithm for Anonymous Agents }
\label{sec:DeterAnonymous}

In this section, we consider a deterministic algorithm to solve the $g$-partial gathering problem for anonymous agents. At first, we show that there exist unsolvable initial configurations in this
model. Later, we propose a deterministic algorithm that solves the $g$-partial gathering problem in $O(kn)$ total moves for any solvable initial configuration.

\subsection{Existence of Unsolvable Initial Configurations}
To explain unsolvable initial configurations, we  define the {\em distance sequence} of a configuration. For configuration $c$, we define the distance sequence of agent $a_h$ as
$D_h(c)=(d^h_0(c),\ldots,d^h_{k-1}(c))$, where $d^h_i(c)$ is the distance between the $i$-th forward agent of $a_h$ and the $(i+1)$-th forward agent of $a_h$ in $c$.
Then, we define the distance sequence of configuration $c$ as the lexicographically minimum sequence among $\{D_h(c)|a_h\in A\}$,
and we denote it by $D(c)$.
In addition, we define several functions and variables for 
sequence $D=(d_0,d_1,\ldots,d_{k-1})$. 
Let $\textit{shift}(D,x)=(d_{x},d_{x+1},\ldots,d_{k-1},d_0,d_1,\ldots,d_{x-1})$ and 
when  $D=\textit{shift}(D,x)$ holds for some $x$ such that $0<x<k$ holds, we say $D$ is \textit{periodic}.
Moreover, we define $period$ of $D$ as the minimum (positive) value such that 
$\textit{shift}(\textit{D},period)=\textit{D}$ holds.

Then, we have the following theorem.
\begin{theorem}
\label{unsolvableRing}
Let $c_0$ be an initial configuration. If $D(c_0)$ is periodic and  $period$ is less than $g$, the $g$-partial gathering problem is not solvable.
\end{theorem}

\begin{proof}

Let $m=k/period$. Let $A_j$ ($0\le j\le period-1$) be a set of agents $a_h$ such that $D_h(c_0)=\textit{shift}(D(c_0),j)$ holds. Then, when all agents move in the synchronous manner, all agents in $A_j$ continue to do
the same behavior and thus they cannot break the periodicity of the initial configuration. Since the number of agents in $A_j$ is $m$ and no two agents in $A_j$ stay at the same node, there exist $m$
nodes where agents stay in the final configuration. However, since $k/m=period<g$ holds, it is impossible that at least $g$ agents meet at the $m$ nodes. Therefore, the $g$-partial gathering problem is
not solvable.
\end{proof}

\subsection{Proposed Algorithm}

In this section, we propose a deterministic algorithm to solve the $g$-partial gathering problem in $O(kn)$ total moves
for solvable initial configurations. Let $D=D(c_0)$ be the distance sequence of
initial configuration $c_0$. 
From Theorem \ref{unsolvableRing},  the $g$-partial gathering problem is not solvable if $period<g$.
On the other hand, our proposed algorithm solves the $g$-partial gathering problem if $period\geq g$ holds.
 In this section, we assume that each agent knows the
number $k$ of agents.

The idea of the algorithm is as follows: 
First each agent $a_h$ travels once around  the ring and 
obtains the distance sequence $D_h(c_0)$. After that, $a_h$ computes $D$ and $period$. If $period<g$ holds, $a_h$
terminates the algorithm because the $g$-partial gathering problem is not solvable. Otherwise, agent $a_h$ identifies nodes such that agents in $\{a_\ell|D=D_\ell(c_0)\}$ initially exist. Then, $a_h$
moves to the nearest node among them. Clearly $period\,(\ge g)$ agents meet at the node, and the algorithm solves the $g$-partial gathering problem.

\begin{algorithm}[t!]
\caption{The behavior of active agent $a_h$ ($v_j$ is the current node of $a_h$.)}
\label{AnonymousDeterRing}
\begin{algorithmic}[1]
\item [\textbf{Variables in Agent $a_h$}]
\item [int $a_h.\textit{total};$]
\item [int $a_h.dis;$]
\item [int $a_h.x;$]
\item [array of int] $a_h.D$[ ];
\item [array of int] $D_{min}$[ ];
\item[\textbf{Main Routine of Agent $a_h$}]
\STATE $a_h.\textit{total}=0$
\STATE $a_h.dis=0$
\WHILE {$a_h.\textit{total}\neq k$}
\STATE move to the forward node
\WHILE {$v_j.initial = false $}
\STATE move to the forward node
\STATE  $a_h.dis=a_h.dis+1$
\ENDWHILE
\STATE  $a_h.D[a_h.\textit{total}]=a_h.dis$
\STATE  $a_h.\textit{total}=a_h.\textit{total}+1$
\STATE  $a_h.dis=0$
\ENDWHILE
\STATE let $D_{min}$ be a lexicographically minimum sequence among $\{\textit{shift}(a_h.D,x)|0\le x\le k-1\}$.
\STATE  $period=\min \{x>0|shift(D_{min},x)=D_{min}\}$
\IF {$  (g>period)$}
\STATE terminate the algorithm
\STATE // the $g$-partial gathering problem is not solvable
\ENDIF
\STATE $a_h.x= \min \{x\le 0| \textit{shift}(a_h.D,x) = D_{min}\}$ 
\STATE move to the forward node $\sum^{a_h.x-1}_{i=0}{a_h.D[i]}$ times
\end{algorithmic}

\end{algorithm}

We have the following theorem about Algorithm \ref{AnonymousDeterRing}.

\begin{theorem}

When agents have no IDs,
our deterministic algorithm solves the $g$-partial gathering problem in $O(kn)$ total moves
if the initial configuration is solvable.

\end{theorem}

\begin{proof}
At first, we show the correctness of the algorithm.
Each agent $a_h$ moves around the ring, and computes the distance sequence $D_{min}$ and its  $period$. If $period<g$ holds, the $g$-partial gathering problem is not solvable from Theorem
\ref{unsolvableRing} and $a_h$ terminates the algorithm.
In the following, we consider the case that $period\geq g$ holds.
From line 20 in Algorithm \ref{AnonymousDeterRing}, each agent moves to the forward node $\sum^{a_h.x-1}_{i=0}{a_h.D[i]}$ times.
By this behavior, each agent $a_h$ moves to the nearest node such that agent $a_\ell$ with $a_\ell.D=D(c_0)$ initially exists. Since $period \,(\ge g)$ agents move to the node, the algorithm solves
the $g$-partial gathering problem.

Next, we analyze the total moves required to solve the $g$-partial gathering problem. In Algorithm \ref{AnonymousDeterRing}, all agents circulate the ring. 
This requires $O(kn)$ total moves.
After this, each agent moves at most $n$ times to meet other agents. This requires $O(kn)$ total moves.
Therefore, agents solve the $g$-partial gathering problem in $O(kn)$ total moves.
\end{proof}

\section{Conclusion}
\label{conclusion}

In this paper, we have proposed three  algorithms to solve the $g$-partial gathering problem in asynchronous unidirectional rings.
The first algorithm is deterministic and works for  distinct agents.
The second algorithm is randomized and works for  anonymous agents under the assumption that each agent knows the total number of agents.
The third algorithm is deterministic and works for anonymous agents under the assumption that each agent knows the total number of agents.
In the first and second algorithms, several agents are elected as leaders by executing the leader agent election partially.
The first algorithm uses agents' distinct IDs and the second algorithm uses random IDs.
In the both algorithms, after the leader election, leader agents instruct the other agents where they meet. 
On the other hand, in the third algorithm, each agent moves around the ring and moves to a node and terminates so that at least $g$ agents should meet at the same node.
 We have showed that the first and second  algorithms requires $O(gn)$ total moves, which is asymptotically optimal. 
The future work is to analyze the lower bound under the assumption that the algorithm is deterministic and each agent is anonymous.
We conjecture that it is $\Omega (kn)$, and if the conjecture is correct, we can show that the third algorithm is asymptotically optimal in terms of total moves. 
%propose algorithms to solve  the $g$-partial gathering problem in other network models.

\bibliographystyle{unsrt} 
\bibliography{ref}





%\begin{thebibliography}{00}


%\bibitem{}

%\end{thebibliography}
\end{document}